\documentclass{article}
\usepackage{amsmath,amssymb,amsthm}
\usepackage{stmaryrd}
\usepackage{paralist}
\usepackage{graphicx}
\usepackage{listings}
\usepackage{xspace}
\usepackage{authblk}

\usepackage{color}

\newcommand{\ax}{\ensuremath{\text{ax}}}
\newcommand{\rul}{\ensuremath{\delta}}
\newcommand{\sem}[1]{\llbracket #1 \rrbracket}
\newcommand{\AlphIn}{\ensuremath{\Sigma}}
\newcommand{\TreesIn}{\ensuremath{\mathcal{T}_{\AlphIn}}}
\newcommand{\AlphOut}{\ensuremath{\Delta}}
\newcommand{\dom}{{\sf dom}}
\newcommand{\lcp}{{\sf lcp}}
\newcommand{\lcs}{{\sf lcs}}
\newcommand{\size}[1]{\ensuremath{|#1|}}
\newcommand{\shiftPeriod}[2]{\ensuremath{s(#1,#2)}}
\newcommand{\pseudoShift}[2]{\ensuremath{s'(#1,#2)}}
\newcommand{\shiftWord}[2]{\ensuremath{\rho_{#2}\left[#1\right]}}
\newcommand{\ltw}{\textsc{ltw}\xspace}
\newcommand{\ltws}{\textsc{ltw}s\xspace}
\newcommand{\stw}{\textsc{stw}\xspace}
\newcommand{\stws}{\textsc{stw}s\xspace}

\newcommand{\ruleltw}{\ensuremath{q,f \rightarrow u_0 q_1(x_{\sigma(1)}) \dots q_n(x_{\sigma(n)}) u_n}}
\newtheorem{definition}{Definition}
\newtheorem{corollary}{Corollary}

\newtheorem{theorem}{Theorem}
\newtheorem{lemma}{Lemma}
\newtheorem{example}{Example}
\newtheorem{algo}{Algorithm}{\bfseries}{\itshape}
\newtheorem*{thm}{Theorem}{\bfseries}{\itshape}
\newtheorem*{lem}{Lemma}{\bfseries}{\itshape}

\title{Deciding Equivalence of Linear Tree-to-Word Transducers in Polynomial Time}
\author[1]{Adrien Boiret\thanks{This work was partially supported by a grant
from CPER Nord-Pas de Calais/FEDER DATA Advanced data science and technologies 2015-2020}}
\author[2]{Raphaela Palenta}
\affil[1]{CRIStAL, University Lille 1, France}
\affil[2]{Department of Informatics, Technical University of Munich, Germany}
\date{}

\begin{document}

\setlength{\abovedisplayskip}{0pt}
\setlength{\belowdisplayskip}{0pt}
\setlength{\abovedisplayshortskip}{0pt}
\setlength{\belowdisplayshortskip}{0pt}

 \maketitle

\begin{abstract}
We show that the equivalence of deterministic linear top-down tree-to-word transducers is 
 decidable in polynomial time.  
 Linear tree-to-word transducers are non-copying but not necessarily order-preserving 
 and can be used to express XML and other document transformations. 
 The result is based on a partial normal form that
 provides a basic characterization of 
 the languages produced by linear tree-to-word transducers.
\end{abstract}

\section{Introduction}

Tree transformations are widely used in functional programming 
and document processing. Tree transducers are a general model for transforming structured data like a database
in a structured or even unstructured way.
Consider the following internal representation of a client database that 
should be transformed to a table in HTML.
\begin{center}
\includegraphics[scale = 0.8]{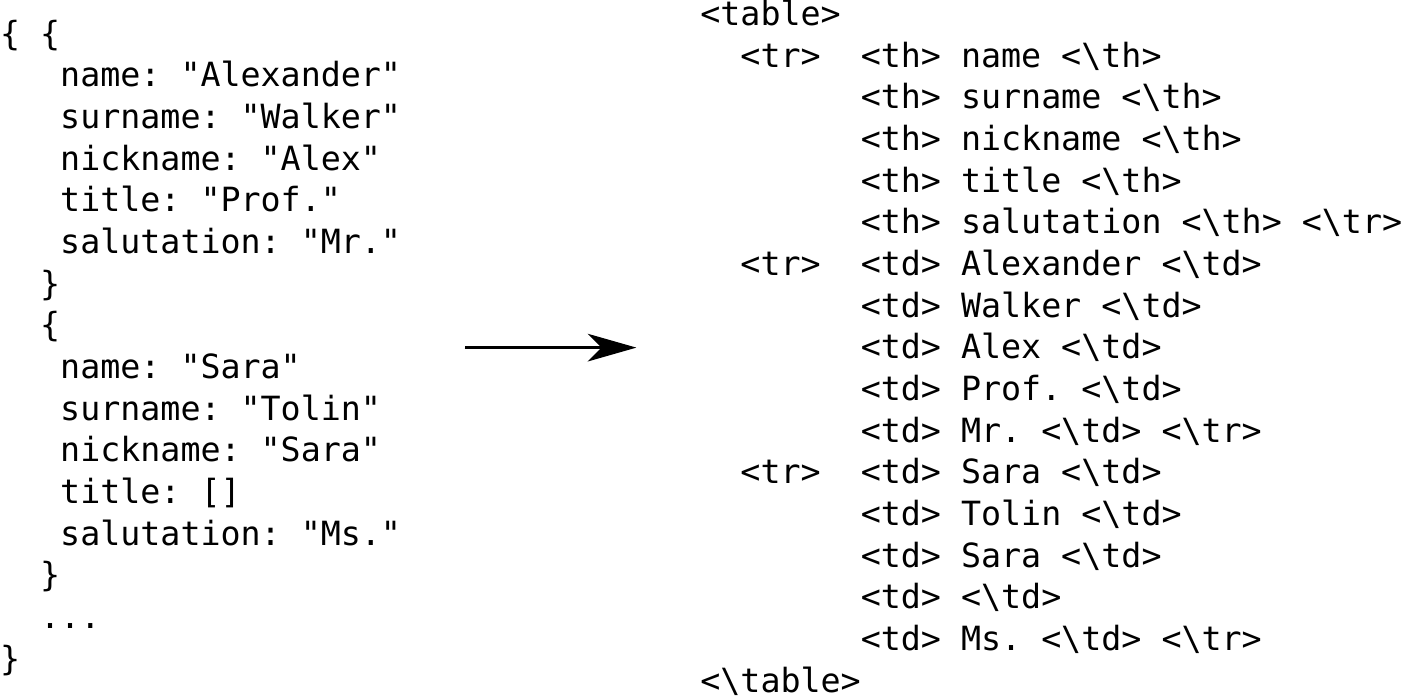}
\end{center}

Top-down tree transducers can be seen as functional programs
that transform trees from the root to the leaves with finite memory.
Transformations where the output is not produced in a structured way or 
where, for example, the output is a string, can be modeled by tree-to-word transducers.

In this paper, we study deterministic linear tree-to-word transducers (\ltws),
a subset of deterministic tree-to-word transducers that are non-copying,
but not necessarily order-preserving. 
Processing the subtrees in an arbitrary order is important to avoid reordering of the 
internal data for different use cases. 
In the example of the client database the names may be needed in different formats, e.g.
\begin{lstlisting}
<salutation> <name> <surname>
<surname>, <name>
<title> <surname>
<title> <surname>, <name>
\end{lstlisting}

The equivalence of unrestricted tree-to-word transducers was a long standing open problem that 
was recently shown to be decidable~\cite{Seidl2015}. 
The algorithm by~\cite{Seidl2015} provides an co-randomized polynomial algorithm for linear transducers.
We show that the equivalence of \ltws is decidable in polynomial time and 
provide a partial normal form. 

To decide equivalence of \ltws, we start in Section~\ref{sec_transducer}
by extending the methods used
for sequential (linear and order-preserving) tree-to-word transducers (\stws), discussed in~\cite{Staworko2009}.
The equivalence for these transducers is decidable in polynomial time~\cite{Staworko2009}.
Moreover a normal form for sequential and linear tree-to-word transducers, 
computable in exponential time, is known~\cite{Laurence2011,Boiret2016}.
Two equivalent \ltws do not necessarily transform their trees
in the same order. However, the differences that can occur are
quite specific and characterized in~\cite{Boiret2016}. We show how they can be identified.
We use the notion of \emph{earliest} states, inspired
by the existing notion of earliest sequential transducers~\cite{Laurence2011}.
In this earliest form, two equivalent \stws can transform subtrees
in different orders only if they fulfill specific properties pertaining
to the periodicity of the words they create.
Computing this normal form is exponential in complexity as the number of 
states may increase exponentially.
To avoid this size increase 
we do not compute these earliest transducers fully,
but rather locally. This means we transform two \ltws with different orders to 
a \emph{partial normal form} in polynomial time (see Section~\ref{sec_PNF})
where the order of their transformation of the different subtrees are the same.
\ltws that transform the subtrees of the input in the same order can be reduced to 
sequential tree-to-word transducers as the input trees can be reordered according 
to the order in the transformation.

A short version of this paper will be published in the proceedings of 
the 20th International Conference on Developements in Language Theory (DLT 2016).

\textbf{Related Work.}
Different other classes of transducers, such as tree-to-tree transducers~\cite{Engelfriet1978}, 
macro tree transducers~\cite{DBLP:journals/jcss/EngelfrietV85}
or nested-word-to-word transducers~\cite{Staworko2009} have been studied. 
Many results for tree-to-tree transducers are known, e.g.\ deciding equivalence~\cite{DBLP:conf/planX/ManethS07},
minimization algorithms~\cite{DBLP:conf/planX/ManethS07} and Gold-style learning algorithms~\cite{DBLP:conf/pods/LemayMN10}.
In contrast, transformations where the output is not generated in a structured way like 
a tree are not that well understood. In macro-tree transducers, the decidability of equivalence
is a well-known and long-standing question~\cite{Engelfriet1980}.
However, the equivalence of linear size increase macro-tree transducers 
that are equivalent to MSO definable transducers 
is decidable~\cite{DBLP:journals/siamcomp/EngelfrietM03,Engelfriet2006}.

\section{Preliminaries}\label{sec_prelim}

Let $\Sigma$ be a ranked alphabet with $\Sigma^{(n)}$ the symbols of rank $n$.
Trees on $\AlphIn$ ($\TreesIn$) are defined inductively: if $f \in \Sigma^{(n)}$,
and $t_1, ... ,t_n\in\TreesIn$, then $f(t_1, ... ,t_n)\in\TreesIn$ is a tree.
Let $\AlphOut$ be an alphabet.
An element $w \in \AlphOut^*$ is a word. For two words $u, v$ we denote the concatenation of these 
two words by $uv$. The length of a word $w$ is denoted by $\size{w}$. 
We call $\varepsilon$ the empty word. We denote $a^{-1}$ the inverse of a symbol $a$ where 
$aa^{-1} = a^{-1}a = \varepsilon$. The inverse of a word $w = u_1 \dots u_n$ is $w^{-1} = u_n^{-1} \dots u_1^{-1}$.

A \emph{context-free grammar} (CFG) is defined as a tuple
$(\AlphOut,N,S,P)$, where $\AlphOut$ 
is the alphabet of $G$, $N$ is a finite set of \emph{non-terminal symbols}, 
$S\in N$ is the initial non-terminal of $G$, $P$ is a finite set of rules of form $A\rightarrow w$,
where $A\in N$ and $w\in(\AlphOut \cup N)^*$. 
A CFG is deterministic if each non-terminal has at most one rule. 

We define the language $L_G(A)$ of a non-terminal $A$
recursively:
if $A\rightarrow u_0A_1u_1 ... A_nu_n$ is a rule of $P$,
with $u_i$ words of $\AlphOut^*$ and $A_i$ non-terminals of $N$,
and $w_i$ a word of $L_G(A_i)$,
then $u_0w_1u_1 ... w_n u_n$ is a word of $L_G(A)$.
We define the context-free language $L_G$ of a context-free grammar $G$ as $L_G(S)$.

A \emph{straight-line program} (SLP) is a deterministic CFG that produces exactly one word.
The word produced by an SLP $(\AlphOut,N,S,P)$ is called $w_S$.

We denote the \emph{longest common prefix of all words} of a language $L$ by $\lcp(L)$. Its \emph{longest common suffix} is $\lcs(L)$.

A word $u$ is said to be \emph{periodic} of period $w$
if $w$ is the smallest word such that $u\in w^*$.
A language $L$ is said to be \emph{periodic} of period $w$
if $w$ is the smallest word such that $L\subseteq w^*$.

A language $L$ is \emph{quasi-periodic} on the left (resp. on the right)
of handle $u$ and period $w$ if $w$ is the smallest word such that 
$L\subseteq uw^*$ (resp. if $L\subseteq w^*u$).
A language is quasi-periodic if it is quasi-periodic on the right or left.
If $L$ is a singleton or empty, it is periodic of period
$\varepsilon$.
Iff $L$ is periodic, it is quasi-periodic
on the left and the right of handle $\varepsilon$. 
If $L$ is quasi-periodic on the left (resp. right) then 
$\lcp(L)$ (resp. $\lcs(L)$) is the shortest word of $L$.

\section{Linear Tree-to-Word Transducers}\label{sec_transducer}

A \emph{linear tree-to-word transducer} (\ltw) is a tuple $M=(\Sigma, \Delta, Q, \ax, \rul)$ where
\begin{compactitem}
\item $\Sigma$ is a ranked alphabet,
\item $\Delta$ is an alphabet of output symbols,
\item $Q$ is a finite set of states,
\item the axiom $\ax$ is of the form $u_0q(x)u_1$,
where $q\in Q$ and $u_0, u_1 \in\AlphOut^*$,
\item $\rul$ is a set of rules of the form 
$\ruleltw$
where $q, q_1, \dots, q_n \in Q$, 
$f \in \AlphIn$ of rank $n$, $u_0, \dots, u_n \in\AlphOut^*$
and $\sigma$ is a permutation from $\{1, \dots, n\}$ to $\{1, \dots, n\}$.
There is at most one rule per pair $q,f$.
\end{compactitem}
The partial function $\sem M_q$ of a state $q$ on an input tree $f(t_1, \dots, t_n)$ 
is defined inductively as 
\begin{compactitem}
\item $u_0 \sem{M}_{q_1}(t_{\sigma(1)}) \dots \sem{M}_{q_n}(t_{\sigma(n)}) u_n,$ 
if $q, f \rightarrow u_0 q_1(x_{\sigma(1)}) \dots q_n(x_{\sigma(n)}) u_n \in \delta$ 
\item undefined, if $q,f$ is not defined in $\delta$.
\end{compactitem}
The partial function $\sem M$ of an \ltw $M$ with axiom $u_0q(x)u_1$
on an input tree $t$ is defined as $\sem M(t)=u_0 \sem M_q(t) u_1$.

Two \ltws $M$ and $M'$ are equivalent if $\sem M = \sem{M'}$.

A \emph{sequential tree-to-word transducer} (\stw) is an \ltw where
for each rule of the form $q,f\rightarrow u_0q_1(x_{\sigma(1)})u_1 \dots q_n(x_{\sigma(n)})u_n$,
$\sigma$ is the identity on $1 \dots n$.

We define \emph{accessibility} of states as the transitive and reflexive closure of appearance in a
rule. This means state  $q$ is accessible from itself, and 
if \ruleltw, and $q$ is accessible from $q'$, then all states $q_i$, $1 \leq i \leq n$,
are accessible from $q'$. 

We denote by $\dom(M)$ (resp. $\dom(q)$) the domain of an \ltw $M$ (resp. a state $q$), i.e.\ all
trees $t\in\TreesIn$ such that $\sem M(t)$ is defined (resp. $\sem M_q(t)$). 
We only consider \ltws with non-empty domains and assume
w.l.o.g. that no state $q$ in an \ltw has an empty domain 
by eliminating transitions using states with empty domain.

We denote by $L_M$ (resp. $L_q$) the range of $\sem M$
(resp. $\sem M_q$), i.e. the set of all images $\sem M(t)$
(resp. $\sem M_q(t)$).
The languages $L_M$ and $L_q$ for each $q\in Q$ are all
context-free languages. We call a state $q$ \emph{(quasi-)periodic} if $L_q$ is (quasi-)periodic.

Note that a word $u$ in a rule of an \ltw can be represented by an
SLP without changing the semantics of the \ltw. 
Therefore a set of SLPs can be added to the transducer and a word 
on the right-hand side of a rule can be represented by an SLPs. 
The decidability of equivalence of \stws in polynomial time still holds true 
with the use of SLPs. 
The advantage of SLPs is that they may compress the size of a word as the following example shows.

\begin{example}\label{ex_expSize}
 We define an SLP $G=(\AlphOut,N,S_0,P)$,
where $N$ is a set $\{S_0,...,S_n\}$, the initial non-terminal is $S_0$,
and $P$ is the set of rules $S_0\rightarrow S_1S_1$,
$S_1\rightarrow S_2S_2$, $\dots$, $S_{n-1}\rightarrow S_nS_n$,
and $S_n\rightarrow a$.
This SLP produces the word $a^{2^{n}}$.
 $G$ has $n+1$ non-terminals and $n+1$ rules. 
Thus, $G$ produces a word that is exponential in the size of $G$.
\end{example}

The results of
this paper require SLP compression to avoid exponential blow-up.
SLPs are used to prevent exponential blow-up in~\cite{Plandowski1995},
where morphism equivalence on context-free languages is decided in polynomial time.

The equivalence problem for sequential tree-to-word transducer
can be reduced to the 
morphism equivalence problem for context-free languages~\cite{Staworko2009}.
This reduction relies on the fact that STWs transform their subtrees in the same order.
As \ltws do not necessarily transform their subtrees in the same order 
the result cannot be applied on \ltws in general.
However, if two \ltws transform their subtrees in the same
order, then the same reduction can be applied.
To formalize that two \ltws transform their subtrees in the same order
we introduce the notion of state co-reachability. 
 Two states $q_1$ and $q_2$ of \ltws $M_1$, $M_2$, respectively, are co-reachable 
if there is an input tree such that the two states are assigned to the same node of the input tree 
in the translations of $M_1$, $M_2$, respectively.

Two \ltws are \emph{same-ordered} if
for each pair of co-reachable states $q_1,\,q_2$ and
for each symbol $f\in\AlphIn$,
neither $q_1$ nor $q_2$ have a rule for $f$,
or if 
$q_1,f \rightarrow u_0 q'_1(x_{\sigma_1(1)}) \dots q'_n(x_{\sigma_1(n)}) u_n$ and 
 $q_2,f \rightarrow v_0 q''_1(x_{\sigma_2(1)}) \dots q''_n(x_{\sigma_2(n)}) v_n$
 are rules of $q_1$ and $q_2$, then $\sigma_1=\sigma_2$.

If two \ltws are same-ordered the input trees can be reordered according to the order in the 
transformations. Therefore for each \ltw a tree-to-tree transducer is constructed that transforms 
the input tree according to the transformation in the \ltw. Then all permutations $\sigma$ in 
the \ltws are replaced by the identity.
Thus the \ltws can be handled as \stws and therefore the equivalence 
is decidable in polynomial time.
\begin{theorem}\label{thm_eqSameOrder}
The equivalence of same-ordered \ltws is decidable in polynomial time.
\end{theorem}

\subsection{Linear Earliest Normal Form}

In this section we introduce the two key properties that are used to 
build a normal form for linear tree-to-word transducers, namely the \emph{earliest} and 
\emph{erase-ordered} properties. The earliest property means that the output is produced 
as early as possible, i.e.\ the longest common prefix (resp. suffix) of $L_q$ 
is produced in the rule in which $q$ occurs, and as left as possible. The erase-ordered property means that 
all states that produce no output are ordered according to the input tree and pushed 
to the right in the rules.

An \ltw is in \emph{earliest form} if
\begin{compactitem}
\item each state $q$ is \emph{earliest}, i.e.\ $\lcp(L_q)=\lcs(L_q)=\varepsilon$,
\item and for each rule $q,f \rightarrow u_0 q_1(x_{\sigma(1)}) \dots q_n(x_{\sigma(n)}) u_n$,
for each $i, 1 \leq i \leq n$, $\lcp(L_{q_i}u_i)=\varepsilon$.
\end{compactitem}

In~\cite[Lemma 9]{Boiret2016} it is shown that for each \ltw $M$
an equivalent earliest \ltw $M'$ can be constructed in exponential time. 
Intuitively, if $\lcp(L_q) = v \neq \varepsilon$ (resp. $\lcs(L_q) = v \neq \varepsilon$) 
then $q'$ is constructed with $L_{q'} = v^{-1} L_q$ (resp. $L_{q'} = L_q v^{-1}$) and $q(x)$ 
is replaced by $v q'(x)$ (resp. $q'(x) v$). If $\lcp(L_q u) = v \neq \varepsilon$ and $v$ 
is a prefix of $u = vv'$ 
then we push $v$ through $L_q$ by constructing $q'$ with $L_{q'} = v^{-1} L_q v$ and replace 
$q(x) u$ by $v q'(x)v'$.

Note that the construction to build the earliest form $M'$ of an
\ltw $M$ creates a same-ordered $M'$.
Furthermore, if a state $q$ of $M$ and a state $q'$ of $M'$
are co-reachable, then $q'$ is an ``earliest'' version of $q$,
where some word $u$ was pushed out of the production of $q$
to make it earliest, and some word $v$ was pushed through
the production of $q$ to ensure that the rules have the right
property:
there exists $u,\,v\in\AlphOut^*$ such that for all
$t\in\dom(q)$, $\sem {M'}_{q'}(t)=v^{-1}u^{-1}\sem {M}_{q}(t)v$.

\begin{theorem}\label{thm_earliestLTW}
For each \ltw an equivalent same-ordered and earliest \ltw can be constructed in exponential time.
\end{theorem}

The exponential time complexity is caused by a potential exponential size increase in the number of 
states as it is shown in the following example.

We call a state $q$ that produces only the empty word, i.e. $L_q = \{\varepsilon\}$, an \emph{erasing state}.
As erasing states do not change the transformation and can occur at any position 
in a rule we need to fix their position for a normal form.

 An \ltw $M$ is \emph{erase-ordered} if for each rule \ruleltw\ in $M$, 
 if $q_i$ is erasing then for all $j \geq i$, $q_j$ is erasing, 
 $\sigma(i) \leq \sigma(j)$ and $u_j = \varepsilon$. 

We test whether $L_q = \{\varepsilon\}$ in polynomial time and then reorder a rule 
according to the erase-ordered property. If an \ltw is earliest it is still 
earliest after the reordering.

\begin{lemma}[extended from{~\cite[Lemma 18]{Boiret2016}}]\label{lem_eraseOrdered}
 For each (earliest) \ltw an equivalent (earliest) erase-ordered \ltw can be constructed in polynomial time. 
 \end{lemma}

\begin{example}
 Consider the rule $q_0, f \rightarrow q_1(x_4) q_2(x_3) q_1(x_2) q_4(x_1)$ 
 where $q_2$ translates trees of the form $f^n(g), n \geq 0$ to $(abc)^n$, 
 $q_4$ translates trees of the form $f^n(g), n \geq 0$ to $(abc)^{2n}$, 
 $q_1$ translates trees of the form $f^n(g), n \geq 0$ to $\varepsilon$.
 Thus the rule is not erase-ordered. We reorder the rule to the 
 equivalent and erase-ordered rule
 $q_0, f \rightarrow q_2(x_3) q_4(x_1) q_1(x_2) q_1(x_4)$.
\end{example}

If two equivalent \ltws are earliest and erase-ordered, then they are not necessarily same-ordered.
For example, the rule $q, f \rightarrow q_4(x_1) q_2(x_3) q_1(x_2) q_1(x_4)$ is equivalent 
to the rule in the above example but the two rules are not same-ordered.
However, in earliest and erase-ordered \ltws,
we can characterize the differences in the orders of
equivalent rules:
Just as two words $u$, $v$ satisfy the equation $uv = vu$
if and only if there is a word $w$ such that 
$u \in w^*$ and $v \in w^*$, the only way for equivalent earliest and erase-ordered \ltws
to not be same-ordered is to switch periodic states.

\begin{theorem}[\cite{Boiret2016}]\label{thm_coreachEquiv}
Let $M$ and $M'$ be two equivalent erase-ordered and earliest \ltws and
$q$, $q'$ be two co-reachable states in $M$, $M'$, respectively. Let

\centerline{$q,f \rightarrow u_0 q_1(x_{\sigma_1(1)}) \dots q_n(x_{\sigma_1(n)}) u_n$ and 
$q',f \rightarrow v_0 q'_1(x_{\sigma_2(1)}) \dots q'_n(x_{\sigma_2(n)}) v_n$}

\noindent be two rules for $q$, $q'$. Then 
\begin{compactitem}
  \item for $k < l$ such that $\sigma_1(k) = \sigma_2(l)$, all $q_i$, $k \leq i \leq l$, are periodic of the same period 
 and all $u_j = \varepsilon$, $k \leq j < l$, 
 \item for $k, l$ such that 
$\sigma_1(k) = \sigma_2(l)$, $\sem M_{q_k}=\sem {M'}_{q'_{l}}$.
\end{compactitem}

\end{theorem}

As the subtrees that are not same-ordered in two equivalent earliest and erase-ordered states 
are periodic of the same period the order of these can be changed without changing the semantics.
Therefore the order of these subtrees can be fixed such that equivalent earliest and erase-ordered 
\ltws are same-ordered. Then the equivalence is decidable in polynomial time, see Theorem~\ref{thm_eqSameOrder}. 
However, building the earliest form of an \ltw is in exponential time. 

To circumvent this difficulty, we will show that the first part of Theorem~\ref{thm_coreachEquiv} still holds even on a
\emph{partial normal form}, where only
quasi-periodic states are earliest and the longest common prefix of 
parts of rules $q(x) u$ with $L_{q} u$ being quasi-periodic is the empty word.

\begin{theorem}\label{thm_periodic}
 Let $M$ and $M'$ be two equivalent erase-ordered \ltws such that 
 \begin{compactitem}
  \item all quasi-periodic states $q$ are earliest, i.e.\ $\lcp(q) = \lcs(q) = \varepsilon$
  \item for each part $q(x)u$ of a rule where $L_{q} u$ is quasi-periodic, $\lcp(L_{q} u) = \varepsilon$ 
 \end{compactitem}
 Let $q$, $q'$ be two co-reachable states in $M$, $M'$, respectively and
 
 \centerline{$q,f \rightarrow u_0 q_1(x_{\sigma_1(1)}) \dots q_n(x_{\sigma_1(n)}) u_n$  and 
$q',f \rightarrow v_0 q'_1(x_{\sigma_2(1)}) \dots q'_n(x_{\sigma_2(n)}) v_n$}

\noindent be two rules for $q$, $q'$. Then 
  for $k < l$ such that $\sigma_1(k) = \sigma_2(l)$, all $q_i$, $k \leq i \leq l$, are periodic of the same period 
 and all $u_j = \varepsilon$, $k \leq j < l$.
\end{theorem}

\section{Partial Normal Form}\label{sec_PNF}

In this section we introduce a partial normal form
for \ltws that does not suffer from the exponential blow-up of the
earliest form. 
Inspired by Theorem~\ref{thm_periodic},
we wish to solve order differences by switching adjacent
periodic states of the same period. 
Remember that the earliest form of a state $q$ is constructed by removing the longest common 
prefix (suffix) of $L_q$ to produce this prefix (suffix) earlier. 
It follows that all non-earliest states from which $q$ can be 
constructed following the earliest form are quasi-periodic.

We show that building the earliest form of a quasi-periodic state 
or a part of a rule $q(x) u$ with $L_q u$ being quasi-periodic is
in polynomial time. Therefore building the following partial normal form is in polynomial time.

\begin{definition}\label{def_partialnormalform}
 A linear tree-to-word transducer is in \emph{partial normal form} if
 \begin{compactenum}
  \item all quasi-periodic states are earliest,
  \item it is erase-ordered and 
  \item for each rule \ruleltw\ if  
  $L_{q_i} u_i L_{q_{i+1}}$ is quasi-periodic
  then $q_i(x_{\sigma(i)}) u_i q_{i+1}(x_{\sigma(i+1)})$ is earliest and 
    $\sigma(i) < \sigma(i+1)$.
 \end{compactenum}
\end{definition}

\subsection{Eliminating Non-Earliest Quasi-Periodic States}
In this part, we show a polynomial time algorithm to build an earliest form of a quasi-periodic state. 
From which an equivalent \ltw can be constructed in polynomial time 
such that any quasi-periodic state is earliest, i.e.\ $\lcp(L_q)=\lcs(L_q)=\varepsilon$.
Additionally, we show that the presented algorithm can be adjusted to
test if a state is quasi-periodic in polynomial time.

As quasi-periodicity on the left and on the right are symmetric properties 
we only consider quasi-periodic states 
of the form $u w^*$ (quasi-periodic on the left). 
The proofs in the case $w^* u$ are symmetric and therefore 
omitted here. 
In the end of this section we shortly discuss the introduced algorithms for the symmetric case $w^* u$.

To build the earliest form of a quasi-periodic state 
we use the property that each state accessible from a quasi-periodic state is as well quasi-periodic.
However, the periods can be shifted as the following example shows.

\begin{example}\label{ex_qpState}
  Consider states $q$, $q_1$ and $q_2$ with rules 
 $q, f \rightarrow a q_1(x_1) c$, $q_1, f \rightarrow aa q_2(x_1) ab$,
 $q_2, f \rightarrow q_2(x_1) abc$, $q_2, g \rightarrow abc$.
 State $q$ accepts trees of the form $f^n(g)$, $n \geq 2$, and produces the language $aaa(abc)^n$, i.e.\ 
 $q$ is quasi-periodic of period $abc$. State $q_1$ accepts trees of the 
 form $f^n(g)$, $n \geq 1$, and produces the language $aa(abc)^{n}ab$, i.e.\ 
 $q_1$ is quasi-periodic of period $cab$. 
 State $q_2$ accepts trees of the form $f^n(g)$, $n \geq 0$ and produces the language $(abc)^{n+1}$, i.e.\ 
 $q_2$ is (quasi-)periodic of period $abc$. 
\end{example}

We introduce two definitions to measure the shift of periods.
We denote by $\shiftWord un$ the \emph{from right-to-left shifted word of $u$ of shift $n$, $n \leq \size{u}$}, i.e.\ 
$\shiftWord un = u'^{-1}uu'$ where $u'$ is the prefix of $u$ of size $n$. 
If $n \geq |u|$ then $\rho_n[u] = \rho_{m}[u]$ with $m = n \mod |u|$.

For two quasi-periodic states $q_1,\,q_2$ of period $u=u_1u_2$ and $u'=u_2u_1$, respectively, 
we denote the \emph{shift in their period} by $\shiftPeriod{q_1}{q_2}=\size{u_1}$.

The size of the periods of a quasi-periodic state and the 
states accessible from this state can be computed from the size of the shortest words of the languages produced 
by these states.

\begin{lemma}\label{lem_quasi_period}
If $q$ is quasi-periodic on the left with period $w$, and $q'$ accessible from $q$,
then $q'$ is quasi-periodic with period $\varepsilon$ or a shift of $w$. 
Moreover we can calculate the shift $\shiftPeriod{q}{q'}$ in polynomial time.
\end{lemma}

We now use these shifts to build, for a state $q$ in $M$ that is quasi-periodic on the left, a transducer
$M^q$ equivalent to $M$ where each occurrence of $q$ is replaced by 
its equivalent earliest form, i.e.\ a periodic state and the corresponding prefix.

\begin{algo}\label{earliestAlgo}

Let $q$ be a state in $M$ that is quasi-periodic on the left.
$M^q$ starts with the same states, axiom, and rules as $M$.
\begin{itemize}
\item For each state $p$ accessible from $q$,
we add a copy $p^e$ to $M^q$.
\item For each rule $p,f \rightarrow u_0 q_1(x_{\sigma(1)})\ldots q_n(x_{\sigma(n)}) u_n$ in $M$ with 
$p$ accessible from $q$, we add a rule 
 $p^e,f\rightarrow u_p q_1^e(x_{\sigma(1)})q_2^e(x_{\sigma(2)})\ldots q_n^e(x_{\sigma(n)})$ with 
 $u_p = \shiftWord {\lcp(p)^{-1}u_0\lcp(q_1) \dots \lcp(q_n) u_n}{\shiftPeriod{q}{p}}$
in $M^q$.
\item We delete state $q$ in $M^q$ and replace any occurrence of $q(x)$
in a rule or the axiom of $M^q$ by $\lcp(q)q^e(x)$.
\end{itemize}
\end{algo}
Note that $\lcp(p)^{-1}u_0\lcp(q_1) \dots \lcp(q_n) u_n$ is equivalent to deleting 
the prefix of size $\size{\lcp(p)}$ from the word $u_0\lcp(q_1) \dots \lcp(q_n) u_n$.

Intuitively, to build the earliest form of a state $q$ that is quasi-periodic on the left 
we need to push all words and all longest common prefixes of states on the right-hand side of 
a rule of $q$ to the left. Pushing a word to the left through a state needs to shift the 
language produced by this state. 
We explain the algorithm in detail on state $q$ from Example \ref{ex_qpState}.

\begin{example}\label{ex_earliest}

Remember that $q$ produces the language $aaa(abc)^n, n \geq 2$ and $q_1$, $q_2$ accessible 
from $q$ produce languages $aa(abc)^nab, n \geq 1$ and $(abc)^{n+1}, n \geq 0$, respectively.
Therefore $\lcp(q) = aaaabcabc$, $\lcp(q_1) = aaabcab$ and $\lcp(q_2) = abc$. 
We start with state $q$. 
As there is only one rule for $q$ the longest common prefix of $q$ and the longest common prefix of this rule are the same and 
therefore eliminated. 

$\begin{aligned}
q^e, f &\rightarrow \rho_{s(q,q)}[\lcp(q)^{-1} a \lcp(q_1) c]q_1^e(x_1) \\
    &\rightarrow \rho_{s(q,q)}[(aaaabcabc)^{-1} a aaabcab c]q_1^e(x_1) \\
    &\rightarrow q_1^e(x_1)
\end{aligned}$

\noindent As there is only one rule for $q_1$ the argumentation is the same and we get $q_1^e, f \rightarrow q_2^e$.
For the rule $q_2, f$ we calculate the longest common prefix of the right-hand side $\lcp(q_2) abc = abcabc$ 
that is larger than the longest common prefix of $q_2$.
Therefore we need to calculate the shift $s(q, q_2) = s(q, q_1) + s(q_1, q_2) = |c| + |ab| = 3$ as 
$q_1$ is accessible from $q$ in rule $q, f$ and $q_2$ is accessible from $q_1$ in rule $q_1, f$. 
This leads to the following rule.

$\begin{aligned}
 q_2^e, f &\rightarrow \rho_{s(q, q_2)}[\lcp(q_2)^{-1} \lcp(q_2) abc] q_2^e(x_1) \\
      &\rightarrow \rho_{3}[(abc)^{-1} abcabc] q_2^e(x_1) \\
      &\rightarrow abc q_2^e(x_1)
\end{aligned}$

\noindent As the longest common prefix of $q_2$ is the same as the longest common prefix of the 
right-hand side of rule $q_2, g$ we get $q_2^e, g \rightarrow \varepsilon$.
The axiom of $M^q$ is $\lcp(q) q^e(x_1) = aaaabcabc q^e(x_1)$.

\end{example}

\begin{lemma}\label{lem_partial_earliest}
 Let $M$ be an \ltw and $q$ be a state in $M$ that is quasi-periodic on the left. 
 Let $M^q$ be constructed by Algorithm~\ref{earliestAlgo} and $p^e$ be a state in $M^q$ 
  accessible from $q^e$.
 Then $M$ and $M^q$ are equivalent and $p^e$ is earliest.
\end{lemma}

To replace all quasi-periodic states by their equivalent earliest form
we need to know which states are quasi-periodic.
Algorithm~\ref{earliestAlgo} can be modified to test an arbitrary state for 
quasi-periodicity on the left in polynomial time. 
The only difference to
Algorithm~\ref{earliestAlgo} is that we do not know how to compute 
$\lcp(p)$ in polynomial time and $\shiftPeriod{q}{p}$ does not exist.
We therefore substitute $\lcp(p)$ by some smallest word of $L_{p}$
and we define a mock-shift $s'(q,p)$ as follows
\begin{compactitem}
\item $\pseudoShift{q}{q} = 0$ for all $q$,
\item if \ruleltw, we say $\pseudoShift{q}{q_i}=\size{u_i w_{q_{i+1}} \dots w_{q_{n}} u_n}$, 
where $w_q$ is a shortest word of $L_q$,
\item if $\pseudoShift{q_1}{q_2} = n$ and $\pseudoShift{q_2}{q_3}=m$
then $\pseudoShift{q_1}{q_3}=n+m$.
\end{compactitem}
If several definitions of $\pseudoShift{q}{p}$ exist, we use the smallest.
If $p$ is accessible from a quasi-periodic $q$, then
$\pseudoShift{q}{p}=\shiftPeriod{q}{p}$.

\begin{algo}\label{algo_testQP}
Let $M=(\Sigma, \Delta, Q,\,\ax,\,\rul)$ be an \ltw and $q$ be a state in $M$. 
We build an \ltw $T^{q}$ as follows.
\begin{compactitem}
\item For each state $p$ accessible from $q$,
we add a copy $p^e$ to $T^q$.
\item The axiom is $w_q q^e(x)$ where $w_q$ is a shortest word of $L_q$.
\item For each rule $p,f \rightarrow u_0 q_1(x_{\sigma(1)})\ldots q_n(x_{\sigma(n)}) u_n$ in $M$ with 
$p$ accessible from $q$, we add a rule 
$$p^e,f\rightarrow u_{p}
q_1^e(x_{\sigma(1)})q_2^e(x_{\sigma(2)})\ldots q_n^e(x_{\sigma(n)})$$
in $T^q$, where $u_{p}$ is constructed as follows.
\begin{compactitem}
\item We define $u=u_0w_1 \dots w_n u_n$, where $w_i$ is a shortest word of $L_{q_i}$.
\item Then we remove from $u$ its prefix of size $\size{w'}$, where
$w'$ is a shortest word of $L_{p}$. We obtain a word $u'$.
\item Finally, we set $u_{p} = \rho_{\pseudoShift{q}{p}}[u']$.\end{compactitem}
\end{compactitem}
\end{algo}

As the construction of Algorithms~\ref{earliestAlgo} and~\ref{algo_testQP} are the same if 
the state $q$ is quasi-periodic, $\sem{M}_q$ and $\sem{T^q}$ are equivalent if $q$ is quasi-periodic. 
Moreover, $q$ is quasi-periodic if $\sem{M}_q$ and $\sem{T^q}$ are equivalent.

\begin{lemma}\label{lem_proofAlgoQP}
Let $q$ be a state of an \ltw $M$ and $T^q$ be constructed by Algorithm~\ref{algo_testQP}. 
Then $M$ and $T^q$ are same-ordered and $q$ is quasi-periodic on the left if and only if
$\sem {M}_q=\sem {T^q}$ and $q^e$ is periodic.
\end{lemma}

As $M$ and $T^q$ are same-ordered we can test the equivalence in polynomial time, cf. Theorem~\ref{thm_eqSameOrder}.
Moreover testing a CFG for periodicity is in polynomial time and therefore testing 
a state for quasi-periodicity is in polynomial time.

Algorithm~\ref{algo_testQP} can be applied to a part $q(x) u$ of a rule to test 
$L_{q} u$ for quasi-periodicity on the left. In this case  
for each rule 
\ruleltw\ a rule $\hat{q}, f \rightarrow u_0 q_1(x_{\sigma(1)}) \dots q_n(x_{\sigma(n)}) u_n u$ is 
added to $M$ and 
each occurrence of the part $q(x) u$ 
in a rule of $M$ is replaced by $\hat{q}(x)$. We then apply the above algorithm to $\hat{q}$ 
and test $\sem{M}_{\hat{q}}$ and $\sem{T^{\hat{q}}}$ for equivalence and $\hat{q}^e$ for 
periodicity. 
 \begin{example}   Let $q$ be a state with the rules 
  $q, f \rightarrow bca q(x_1)$, $q, g \rightarrow \varepsilon$.
  Thus, $q$ transforms trees of the form $f^n(g)$, $n \geq 0$ to $(bca)^n$. 
  We use Algorithm \ref{algo_testQP} to test $L_{q} bc$ for quasi-periodicity on the left. 
  As explained above we introduce a state $\hat{q}$ with the rules
  $\hat{q}, f \rightarrow bca \hat{q}(x_1)$, $\hat{q}, g \rightarrow bc$. 
  We now apply Algorithm \ref{algo_testQP} on $\hat{q}$. 
  We build $T^{\hat{q}} = \{\{f,g\}, \{a,b,c\}, \{\hat{q}^e\}, ax, \delta\}$ as follows.
  The axiom $ax$ is $bc \hat{q}^e(x_0)$ as the shortest word of $L_{\hat{q}}$ is $bc$. 
  For the rule $\hat{q}, f$ we build $u = bcabc$ as $bc$ is the shortest 
  word of $\hat{q}$. Then we obtain $u' = abc$ and 
  $u_{\hat{q}} = \rho_{s'(\hat{q}, \hat{q})}[abc] = abc$.
  Thus we get $\hat{q}^e, f \rightarrow abc \hat{q}^e(x_1)$. 
  For the rule $\hat{q}, g$ we build $u = bc$ and obtain $u' = \varepsilon$ as 
  the shortest word of $\hat{q}$ is $bc$. Thus we get 
  $\hat{q}^e, g \rightarrow \varepsilon$. 
  
  $T^{\hat{q}}$ transforms trees of the form 
  $f^n(g)$ to $bc(abc)^n$ and $\hat{q}$ transforms trees of the form 
  $f^n(g)$ to $(bca)^n bc$. Thus, they are equivalent.
  Additionally $\hat{q}^e$ is periodic with period $abc$. 
  It follows that $L_{q_1} bc$ is quasi-periodic.
 \end{example}

We introduced algorithms to test states for quasi-periodicity on the left and to 
build the earliest form for such states. These two algorithms can be adapted for 
states that are quasi-periodic on the right. There are two main differences.
First, as the handle is on the right the shortest word of a language $L$ 
that is quasi-periodic on the right 
is $\lcs(L)$. Second, instead of pushing words through a periodic language to the left 
we need to push words through 
a periodic language to the right. 

Hence, we can test each state $q$ of an \ltw $M$ for quasi-periodicity 
on the left and right. 
If the state is quasi-periodic we replace $q$ by its earliest form.
Algorithm~\ref{earliestAlgo} and~\ref{algo_testQP} run in polynomial time 
if SLPs are used. This is crucial as the shortest word 
of a CFG can be of exponential size, cf. Example~\ref{ex_expSize}.
However, the operations that are needed in the algorithms, 
namely constructing the shortest word of a CFG and 
removing the prefix or suffix of a word, are in polynomial time using SLPs, cf.~\cite{Lohrey2014}.

\begin{theorem}\label{lem_noMoreQuasi}
Let $M$ be an \ltw. Then an equivalent \ltw $M'$ where 
all quasi-periodic states are earliest can be 
constructed in polynomial time.
\end{theorem}

\subsection{Switching Periodic States}

In this part we obtain the partial normal form by ordering periodic states 
of an erase-ordered transducer where all quasi-periodic states are earliest. 
Ordering means that if the order of the subtrees in the translation 
can differ, we choose the one similar to the input, i.e.\ 
if $q(x_3)q'(x_1)$ and $q'(x_1)q(x_3)$ are equivalent, we choose the second order.
We already showed how we can build a transducer where each quasi-periodic 
state is earliest and therefore periodic. 
However, we need to make parts of rules earliest such that periodic states 
can be switched as the following example shows.

\begin{example}
 Consider the rule $q,h \rightarrow q_1(x_2) b q_2(x_1)$ where $q_1$, $q_2$ have the rules 
 $q_1, f \rightarrow bcabca q_1(x)$, $q_1, g \rightarrow \varepsilon$, 
 $q_2, f \rightarrow cab q_2(x)$, $q_2, g \rightarrow \varepsilon$. 
 States $q_1$ and $q_2$ are earliest and periodic but not of the same period as 
 a subword is produced in between.
 We replace the non-earliest and quasi-periodic part $q_1(x_2) b$ by their earliest form. This leads to 
 $q, h \rightarrow b q_1^e(x_2) q_2(x_1)$ with 
 $q_1^e, f \rightarrow cabcab q_1^e(x)$, $q_1^e, g \rightarrow \varepsilon$. 
 Hence, $q_1^e$ and $q_2$ are earliest and periodic of the same period and can be switched in the rule.
\end{example}

To build the earliest form of a quasi-periodic part of a rule $q(x) u$ each 
occurrence of this part is replaced by a state $\hat{q}(x)$ and for each rule 
$\ruleltw$ a rule $\hat{q}, f \rightarrow u_0 q_1(x_{\sigma(1)}) \dots q_n(x_{\sigma(n)}) u_n u$ 
is added. Then we apply Algorithm~\ref{earliestAlgo} on $\hat{q}$ to replace $\hat{q}$ and therefore 
$q(x) u$ by their earliest form. Iteratively this leads to the following theorem.

\begin{theorem}\label{earliestPart}
For each \ltw $M$ where all quasi-periodic states are earliest  we can build in polynomial time an equivalent \ltw $M'$ such that each part 
 $q(x) u$ of a rule in $M$ 
 where $L_{q} u$ is quasi-periodic is earliest. 
\end{theorem}

In Theorem \ref{thm_periodic} we showed that order differences in 
equivalent erase-ordered \ltws where all quasi-periodic states are earliest 
and all parts of rules $q(x) u$ are earliest are caused by adjacent periodic states.
As these states are periodic of the same period and no words are produced in between 
these states can be reordered without changing the semantics of the \ltws.

\begin{lemma}\label{lem_switch}
 Let $M$ be an \ltw such that 
 \begin{compactitem}
  \item $M$ is erase-ordered,
  \item all quasi-periodic states in $M$ are earliest and   \item each $q_i(x_{\sigma(i)}) u_i$ in a 
  rule of $M$ that is quasi-periodic is earliest.
 \end{compactitem}
 Then we can reorder adjacent periodic states $q_i(x_{\sigma(i)}) q_{i+1}(x_{\sigma(i+1)})$ 
 of the same period in the 
 rules of $M$ such that $\sigma(i) < \sigma(j)$ in polynomial time. 
 The reordering does not change the transformation of $M$.
\end{lemma}

We showed before how to construct a transducer with the preconditions needed in 
Lemma~\ref{lem_switch} in polynomial time. Note that replacing a 
quasi-periodic state by its earliest form can break the erase-ordered property. 
Thus we need to replace all quasi-periodic states by its earliest form \emph{before} 
building the erase-ordered form of a transducer.
Then Lemma~\ref{lem_switch} is the last step to obtain the partial normal form for an \ltw.

\begin{theorem}\label{thm_normalForm}
 For each \ltw we can construct an equivalent \ltw 
 that is in partial normal form in polynomial time.
\end{theorem}

 \subsection{Testing Equivalence in Polynomial Time}
 
 It remains to show that the equivalence problem of \ltws in partial normal form is decidable in 
 polynomial time.
 The key idea is that two equivalent \ltws in partial normal form are same-ordered. 
 
 Consider two equivalent \ltws $M_1$, $M_2$ where all quasi-periodic states and all parts 
 of rules $q(x) u$ with $L_{q}u$ is quasi-periodic are earliest. In Theorem~\ref{thm_periodic} 
 we showed if the orders $\sigma_1$, $\sigma_2$ of two co-reachable states $q_1$, $q_2$ of 
 $M_1$, $M_2$, respectively, for the same input differ then the states causing this order differences 
 are periodic with the same period. The partial normal form solves this order differences such that 
 the transducers are same-ordered.

\begin{lemma}\label{lem_subtreeorder}
 If $M$ and $M'$ are equivalent and in partial normal form then they are same-ordered.
\end{lemma}

 As the equivalence of same-ordered \ltws is decidable in polynomial time (cf. Theorem~\ref{thm_eqSameOrder}) 
 we conclude the following.
 \begin{corollary}
  The equivalence problem for \ltws in partial normal form is decidable in polynomial time.
 \end{corollary}

 To summarize, the following steps run in polynomial time and transform a \ltw $M$ into its partial normal form.
 \begin{compactenum}
  \item Test each state for quasi-periodicity. If it is quasi-periodic replace the state by its earliest form.
  \item Build the equivalent erase-ordered transducer.
  \item Test each part $q_i(x_i) u_i$ in each rule from right to left for 
  quasi-periodicity on the left. If it is quasi-periodic on the left replace the part 
  by its earliest form.
  \item Order adjacent periodic states of the same period according to the input order.
 \end{compactenum}
 This leads to our main theorem.
\begin{theorem}\label{thm_main}
 The equivalence of \ltws is decidable in polynomial time.
\end{theorem}

\section{Conclusion}

The equivalence problem for linear tree-to-word transducers can be decided in polynomial time.
To prove this we used a reduction to the equivalence problem
between sequential transducers \cite{Laurence2011},
or more exactly, to an extension of this result to same-ordered transducers.
This reduction hinges on two points.
First, we showed that the only structural differences between two
equivalent earliest linear transducers are
caused by periodic languages which are interchangeable. The structural characteristic of periodic languages 
has been used in the normalization of \stws \cite{Laurence2011}.
Second, we showed that if building a fully earliest transducer is
potentially exponential, our reduction only requires
quasi-periodic states to be earliest, which can be done in polynomial time.
The use of the equivalence problem for morphisms on a CFG \cite{Plandowski1995}
and of properties on straight-line programs \cite{Lohrey2012}
is essential here as it was in \cite{Laurence2011,DBLP:conf/lata/LaurenceLNST14}.
This leads to further research questions, starting with generalization
of this result to all tree-to-words transducers.
Furthermore, is it possible that these techniques can be used to
decrease the complexity of some problems in other classes of
transducer classes, such as top-down tree-to-tree transducers,
where the equivalence problem is known to be between \textsc{Exptime}-Hard
and \textsc{NExptime}?

\bibliographystyle{plain}

\begin{appendix}

\section{Proof of Theorem \ref{thm_periodic}}
\begin{thm} Let $M$ and $M'$ be two equivalent erase-ordered \ltws such that 
 \begin{compactitem}
  \item all quasi-periodic states $q$ are earliest, i.e.\ $\lcp(q) = \lcs(q) = \varepsilon$
  \item for each part $q(x)u$ of a rule where $L_{q} u$ is quasi-periodic, $\lcp(L_{q} u) = \varepsilon$ 
 \end{compactitem}
 Let $q$, $q'$ be two co-reachable states in $M$, $M'$, respectively and
 
 \centerline{$q,f \rightarrow u_0 q_1(x_{\sigma_1(1)}) \dots q_n(x_{\sigma_1(n)}) u_n$  and 
$q',f \rightarrow v_0 q'_1(x_{\sigma_2(1)}) \dots q'_n(x_{\sigma_2(n)}) v_n$}

\noindent be two rules for $q$, $q'$. Then 
  for $k < l$ such that $\sigma_1(k) = \sigma_2(l)$, all $q_i$, $k \leq i \leq l$, are periodic of the same period 
 and all $u_j = \varepsilon$, $k \leq j < l$.
\end{thm}\begin{proof}
Let $M_e$ and $M'_e$ be the equivalent earliest transducer of $M$ and $M'$, respectively, such that 
 $M$ and $M_e$ as well as $M'$ and $M'_e$ are same-ordered (cf. Theorem \ref{thm_earliestLTW}). 
 
 Suppose there exists co-reachable (and thus equivalent) states $q^e$
 and $q'^e$ in $M_e$ and $M'_e$, respectively, with rules
 $$q^e, f \rightarrow v_0 q^e_1(x_{\sigma(1)}) \dots q^e_n(x_{\sigma(n)}) v_n, $$
  $$q'^e, f \rightarrow v'_0 q'^e_1(x_{\sigma'(1)}) \dots q'^e_n(x_{\sigma(n)}) v'_{n}$$
 such that $\sigma \neq \sigma'$. 

Let $i$ be the first index such that $\sigma(i) \neq \sigma'(i)$. 
Following Theorem~\ref{thm_coreachEquiv}, we have $j, j'$ such
that $\sigma'(j') = \sigma(i)$ and $\sigma(j) = \sigma'(i)$ and
all $q^e_l$, $i \leq l \leq j$ are periodic with the same period.
 
 Let $q$ and $q'$ be the states in $M$ and $M'$, respectively, from which the co-reachable states 
 $q^e$ and $q'^e$ were constructed with the earliest construction proposed by \cite{Laurence2011}. 
 From the earliest construction it follows that 
 $q$ and $q'$ are co-reachable.
 Since the construction preserves the rule structure, we have:
$$\ruleltw$$
$$q', f \rightarrow u'_0 q'_1(x_{\sigma'(1)}) \dots q'_n(x_{\sigma(n)}) u'_{n}$$

The earliest construction
gives us that for all $l \in \{1, \dots, n\}$, $\sem {M_e}_{q^e_l}(t)=v^{-1}u^{-1}\sem {M}_{q_l}(t)v$ 
for some $u, v \in \AlphOut^*$.
This means that if $q^e_l$ is periodic, then $q_l$ is quasi periodic in its non-earliest form. 
The same is true for all $q'_l$.

However, the first property we supposed of $M$ and $M'$ implies that
all those $q_l$ and $q'_l$ that are quasi-periodic are not only quasi periodic, but periodic.
Consider a part of the rule $q_i(x_{\sigma(i)}) u_i \dots q_j(x_{\sigma(j)})$ 
that is periodic in the earliest form and therefore quasi-periodic in the non-earliest form. 
The first condition gives us that $q_i, \dots, q_j$ are periodic. However, then the words 
$u_i, \dots, u_{j-1}$ are not necessarily empty. 
As the part $q_i(x_{\sigma(i)}) u_i \dots q_j(x_{\sigma(j)})$ is quasi-periodic 
we know that each part $q_k(x_{\sigma(k)}) u_k$, $i \leq k < j$ is quasi-periodic.
Then the second condition of this theorem guarantees 
that the parts $q_k(x_{\sigma(k)}) u_k$, $i \leq k < j$ are not only quasi-periodic, 
but periodic. From which it follows that 
the words $u_i, \dots, u_{j-1}$ are empty.
As the part $q_i(x_{\sigma(i)}) u_i \dots q_j(x_{\sigma(j)})$ is periodic 
and $u_i, \dots, u_{j-1}$ are empty we get that 
$q_i, \dots, q_j$ are periodic of the same period.
The same holds true for states of a  
part of the rule $q'_i(x_{\sigma'(i)}) u'_i \dots q'_j(x_{\sigma'(j)})$ that is 
periodic in the earliest form.
 \end{proof}

\section{Proof of Lemma \ref{lem_quasi_period}}
\begin{lem} If $q$ is quasi-periodic on the left with period $w$, and $q'$ accessible from $q$,
then $q'$ is quasi-periodic with period $\varepsilon$ or a shift of $w$. 
Moreover we can calculate the shift $\shiftPeriod{q}{q'}$ in polynomial time.
\end{lem}\begin{proof}
This is done as an iterative proof with the following elementary
step:
If $\ruleltw$, and $q$ is quasi-periodic on the left with handle $u$
and period $w$,
then for all $i$ between $1$ and $n$, $q_i$ is quasi-periodic
with period $\varepsilon$ or a shift of $w$.

We pick $v_j$ a smallest word
produced by state $q_j$.
We then have that for all $t\in\dom(q_i)$,
$u_0v_1...u_{i-1}\sem M_{q_i}(t)u_{i}...v_nu_n\in L_q$.
If we call $u_l=u_0v_0...u_{i-1}$ and $u_r=u_{i}...v_nu_n$,
we obtain that $L_{q_i}\subseteq u_l^{-1}L_qu_r^{-1}$.
Since $L_{q}\subseteq uw^*$, we can say $L_{q_i}\subseteq u_l^{-1} (uw^*)u_r^{-1}$.
It is a classical result of regular languages that
$(uw^*)u_r^{-1}$ is either empty, a singleton,
or a quasi-periodic language of period $u_rwu_r^{-1}$.
By further removing a prefix to this language the period does not change. 
Hence, we get that
$u_l^{-1}(uw^*)u_r^{-1}$ is also either empty, a singleton,
or a quasi-periodic language of period $u_rwu_r^{-1}$.
This means that $q_i$ is quasi-periodic, of period $\varepsilon$,
or $u_rwu_r^{-1}$, which is a shift of $q$.
The size of $u_r$ can easily be computed from the sizes
of the minimal productions of states $q_j$. 
We build the CFG for $L_{q_j}$.
Then, finding the smallest production of $q_j$
and their size is finding the smallest word of $L_{q_j}$ and
their size, which is a polynomial problem on CFG.

To show that the shifts of the periods can be calculated in polynomial time 
we show that shifts are additive in nature:
If $q_1$ has period $w$, $q_1$ and $q_2$ are of shifted period,
and $q_2$ and $q_3$ are of shifted period, then
$q_1$ and $q_3$ are of shifted period, and
$\shiftPeriod{q_1}{q_3}\equiv\shiftPeriod{q_1}{q_2}+\shiftPeriod{q_2}{q_3}\pmod{\size w}$.

If $q_1$ is of period $w$, then $q_2$ is of period $w_2=w'ww'^{-1}$,
where $w'$ is the suffix of $w$ of size $\shiftPeriod{q_1}{q_2}$.
If $q_2$ is of period $w_2$, then $q_2$ is of period $w_3=w''w_2w''^{-1}$,
where $w''$ is the suffix of $w_2$ of size $\shiftPeriod{q_2}{q_3}$.
We then have that $w_3= (w''w')w(w''w')^{-1}$, where
$(w''w')$ is of size $\shiftPeriod{q_1}{q_2}+\shiftPeriod{q_2}{q_3}$.

We can compute the shift of the period of each state accessible from $q$ rule by rule using
the additive property of the shifts we proved above.
\end{proof}

\section{Proof of Lemma \ref{lem_partial_earliest}}
\begin{lem} Let $M$ be an \ltw and $q$ be a state in $M$ that is quasi-periodic on the left.
 Let $M^q$ be constructed by Algorithm \ref{earliestAlgo} and $p^e$ be a state in $M^q$ 
 accessible from $q^e$.
 Then $M$ and $M^q$ are equivalent and $p^e$ is earliest.
\end{lem}\begin{proof}
To show that $M$ and $M^q$ are equivalent we show that 
$\lcp(q)\sem{M^q}_{q^e}(t) = \sem{M}_q(t)$, for all $t \in \dom(q)$.
To show that $\lcp(q)\sem{M^q}_{q^e}(t) = \sem{M}_q(t)$ 
we show that, for all 
states $p$ accessible from $q$ and all $t \in \dom(p)$,
$\sem {M^q}_{p^e}(t)$ and
$\shiftWord{\lcp(p)^{-1}\sem M_{p}(t)}{\shiftPeriod{q}{p}}$ are equivalent
as then
\begin{align*}
\lcp(q)\sem{M^q}_{q^e}(t) &= \lcp(q) \shiftWord{\lcp(q)^{-1}\sem M_{q}(t)}{\shiftPeriod{q}{q}} \\
&= \lcp(q) \lcp(q)^{-1} \sem M_{q}(t) \\
&= \sem{M}_q(t).
\end{align*}
To show that $\sem {M^q}_{p^e}(t)$ =
$\shiftWord{\lcp(p)^{-1}\sem M_{p}(t)}{\shiftPeriod{q}{p}}$
for all $p$ accessible from $q$ and all $t = (s_1, \dots, s_n) \in \dom(p)$,
we prove that 
$\sem {M^q}_{p^e}(t)$ is of the same period and of the same size as
$\shiftWord{\lcp(p)^{-1}\sem M_{p}(t)}{\shiftPeriod{q}{p}}$.
From Lemma \ref{lem_quasi_period} we know that all $p$ accessible from $q$ are 
quasi-periodic and therefore $\lcp(p)^{-1}\sem M_{p}(t)$ is periodic. 
Hence, if $\sem {M^q}_{p^e}(t)$ and 
$\shiftWord{\lcp(p)^{-1}\sem M_{p}(t)}{\shiftPeriod{q}{p}}$ are 
of the same period and of the same size then they are equivalent.

To show that $\sem {M^q}_{p^e}(t)$ and $\shiftWord{\lcp(p)^{-1}\sem M_{p}(t)}{\shiftPeriod{q}{p}}$ 
have the same size, for all $t \in \dom(p)$,  we show that 
$\size{\sem {M^q}_{p^e}(t)}=\size{\sem M_{p}(t)} - \size{\lcp(p)}.$
The proof is by induction on the input tree. 
For an input tree $t$ with no subtrees we have 
\begin{align*}
\size{\sem {M^q}_{p^e}(t)} &= \size{\rho_{\shiftPeriod{q}{p}}[\lcp(p)^{-1}u_0 \dots u_n]} \\
  &= \size{u_0} + \dots + \size{u_n} - \size{\lcp(p)} \\
  &= \size{\sem M_{p}(t)} - \size{\lcp(p)} \\
  &= \size{\shiftWord{\lcp(p)^{-1}\sem M_{p}(t)}{\shiftPeriod{q}{p}}}.
\end{align*}
Thus, the base casef holds.
Consider an input tree $t = f(s_1, \dots, s_n) \in \dom(p)$.
Then $\sem {M^q}_{p^e}(t)$ is of size
$\size{u_p} + \size{\sem {M^q}_{q_1^e}(s_{\sigma(1)})} + \dots + \size{\sem {M^q}_{q_n^e}(s_{\sigma(n)})}$
with $\size{u_p} = \size{\rho_{\shiftPeriod{q}{p}}[\lcp(p)^{-1}u_0\lcp(q_1) \dots \lcp(q_n) u_n]}$.
Since shifting a word preserves its length, we have
$\size{u_p} = \size{u_0}+\size{\lcp(q_1)}+ \dots + \size{u_n} - \size{\lcp(p)}.$
Thus, we have to show that
$$\size{u_p} + \size{\sem {M^q}_{q_1^e}(s_{\sigma(1)})} + \dots + \size{\sem {M^q}_{q_n^e}(s_{\sigma(n)})} = \size{\shiftWord{\lcp(p)^{-1}\sem M_{p}(t)}{\shiftPeriod{q}{p}}}.$$
By induction we have $\size{\sem {M^q}_{q_i^e}(s_{\sigma(i)})}
=\size{\sem M_{q_i}(s_{\sigma(i)})} - \size{\lcp(q_i)}$.
Thus, we have 
\begin{align*}
 \size{u_p} + &\size{\sem {M^q}_{q_1^e}(s_{\sigma(1)})} + \dots + \size{\sem {M^q}_{q_n^e}(s_{\sigma(n)})} \\
 &= \size{u_0}+\size{\lcp(q_1)}+ \dots + \size{\lcp(q_n)} + \size{u_n} - \size{\lcp(p)} \\
 &\phantom{\null = \size{u_0}}+ \size{\sem M_{q_1}(s_{\sigma(1)})} - \size{\lcp(q_1)} + \dots + \size{\sem M_{q_n}(s_{\sigma(n)})} - \size{\lcp(q_n)} \\
 &= \size{u_0}+\size{\sem M_{q_1}(s_{\sigma(1)})} + \dots + \size{\sem M_{q_n}(s_{\sigma(n)})} + \size{u_n} - \size{\lcp(p)} \\
 &= \size{\sem M_{p}(t)} - \size{\lcp(p)} \\
  &= \size{\shiftWord{\lcp(p)^{-1}\sem M_{p}(t)}{\shiftPeriod{q}{p}}}.
\end{align*}

To show that $\sem {M^q}_{p^e}(t)$ and $\shiftWord{\lcp(p)^{-1}\sem M_{p}(t)}{\shiftPeriod{q}{p}}$ 
have the same period, for all $t \in \dom(p)$, we show that 
$\sem {M^q}_{p^e}(t)\in u^*$ and 
$\shiftWord{\lcp(p)^{-1}\sem M_{p}(t)}{\shiftPeriod{q}{p}} \in u^*$ where $L_q \subseteq wu^*$.
From Lemma \ref{lem_quasi_period} it follows that 
$\shiftWord{\lcp(p)^{-1}\sem M_{p}(t)}{\shiftPeriod{q}{p}} \in u^*$.
To proof that $\sem {M^q}_{p^e}(t)\in u^*$ is
by induction on the input tree.
For an input tree $t$ with no subtrees we have
$\sem{M^q}_{p^e}(t) = \shiftWord{\lcp(p)^{-1}u_0 \dots u_n}{\shiftPeriod{q}{p}}$.
From Lemma \ref{lem_quasi_period} we know that $L_p$ is quasi-periodic of period $u'u''$ 
where $u=u''u'$ and $u'$ is of size $\shiftPeriod{q}{p}$. 
Thus, $\lcp(p)^{-1} u_0 \dots u_n \in (u'u'')^*$ and therefore
$\shiftWord{\lcp(p)^{-1}u_0 \dots u_n}{\shiftPeriod{q}{p}} \in (u'u'')^* = u^*$.
Hence, the base case holds.
Consider an input tree $t = f(s_1, \dots, s_n) \in \dom(p)$.
Then we have $\sem {M^q}_{p^e}(t)= \shiftWord{\lcp(p)^{-1}u_0\lcp(q_1) \dots \lcp(q_n) u_n}{\shiftPeriod{q}{p}} \sem {M^q}_{q_1^e}(s_{\sigma(1)}) \dots \sem {M^q}_{q_n^e}(s_{\sigma(n)})$.
By induction, $\sem {M^q}_{q_i^e}(s_{\sigma(i)})\in u^*$.
With the same argumentation as in the base case 
$\lcp(p)^{-1}u_0\lcp(q_1) \dots \lcp(q_n) u_n \in (u'u'')^*$ with 
$u=u''u'$ and $u'$ is of size $\shiftPeriod{q}{p}$.
Thus, $\shiftWord{\lcp(p)^{-1}u_0\lcp(q_1) \dots \lcp(q_n) u_n}{\shiftPeriod{q}{p}} \in (u''u')^* = u^*$ 
and therefore we get $\sem {M^q}_{p^e}(t) \in u^*$.
\end{proof}

\section{Proof of Lemma \ref{lem_proofAlgoQP}}
\begin{lem}Let $q$ be a state of an \ltw $M$ and $T^q$ be constructed by Algorithm \ref{algo_testQP}. 
Then $M$ and $T^q$ are same-ordered and $q$ is quasi-periodic on the left if and only if
$\sem {M}_q=\sem {T^q}$ and $q^e$ is periodic.
\end{lem}\begin{proof}
 We show that $q$ is quasi-periodic on the left if and only if
 $\sem {M}_q=\sem {T^q}$ and $q^e$ is periodic.
 If $q$ is quasi-periodic on the left the transformation in Algorithm \ref{algo_testQP}
 is the same as in Algorithm \ref{earliestAlgo}. Therefore 
 $\sem {M}_q=\sem {T^q}$ and $q^e$ is periodic.

If $\sem {M}_q=\sem {T^q}$ and $q^e$ is periodic,
then $\sem {M}_q$ is quasi-periodic as $\sem {T^q} = w_q\sem {T^q}_{q^e}$ 
with $w_q$ a shortest word of $L_q$. 
 
 $M$ and $T^q$ are same-ordered as the order of the rules in $T^q$ is 
 the same as in $M$ by construction.
\end{proof}

\section{Proof of Theorem \ref{lem_noMoreQuasi}}
\begin{thm} Let $M$ be an \ltw. Then an equivalent \ltw $M'$ where 
all quasi-periodic states are earliest can be 
constructed in polynomial time.
\end{thm}\begin{proof}
This proof works by induction. We first show that 
if $M = (\AlphIn, \Delta, Q,\,\ax,\,\rul)$ has $n$, $n\geqslant1$  
quasi-periodic states that are non-earliest, then
we can build in polynomial time an equivalent \ltw $M'$ with
$n-1$ non-earliest quasi-periodic states.
Using Algorithm \ref{algo_testQP} we choose $q$ as a non-earliest quasi-periodic state of $Q$.
We apply Algorithm~\ref{earliestAlgo} on state $q$ and get
$M^q$, whose set of state is of form
$Q\sqcup Q^e \backslash \lbrace q\rbrace$,
where $Q^e$ is the set of states $p^e$ with $p$ accessible from $q$ 
that are created by Algorithm~\ref{earliestAlgo}.
According to Lemma~\ref{lem_partial_earliest} all the states of $Q^e$
are periodic. This means that the non-earliest quasi-periodic states
of $M^q$ are all in $Q\backslash \lbrace q\rbrace$.
Since $Q$ has $n$ non-earliest quasi-periodic states,
including $q$, $M^q$ has $n-1$.

Now we can build $M_1$ equivalent to
$M$ with $n-1$ non-earliest quasi-periodic states,
then $M_2$ equivalent to
$M_1$ (hence to $M$) with $n-2$ non-earliest quasi-periodic states,
and so on. Finally we get $M_n$ equivalent to $M$ with no non-earliest quasi-periodic state.
Each step is in polynomial time and the number $n$ is smaller than
the number of states in $M$. 
For each occurence of a state on the right-hand side of a rule there is at most 
one new state needed in the construction. 
Therefore the size increase of the transducer is only polynomial.
To avoid the construction of equivalent states $q$ should be considered before 
$q'$ if $q'$ is accessible from $q$. If $q$ is accessible from $q'$ and 
$q'$ is accessible from $q$ then $q$ is considered first if there is a 
acyclic way from the axiom to $q$ that contains $q'$.

In the above proof we assumed that Algorithm \ref{algo_testQP} and 
\ref{earliestAlgo} run in polynomial time. 
In both algorithms it is crucial that SLPs are used to represent the 
shortest words of the languages produced by the states of a transducer 
as these can be of exponential size, cf. Example \ref{ex_expSize}.
Instead of these uncompressed words nonterminals representing these words 
as SLPs are inserted in the transducers.
All operations that are needed in the algorithms, 
namely constructing a SLP for the shortest word of an CFG, concatenation of SLPs, 
shifting the word produced by an SLP and 
removing the prefix or suffix of an SLP are in polynomial time~\cite{Lohrey2014}.

\end{proof}

\section{Proof of Theorem \ref{earliestPart}}
\begin{thm}
 For each \ltw $M$ where all quasi-periodic states are earliest  we can build in polynomial time an equivalent \ltw $M'$ such that each part 
 $q(x) u$ of a rule in $M$ 
 where $L_{q} u$ is quasi-periodic is earliest.
\end{thm}
\begin{proof}
First, we show that one part $q(x) u$ of a rule where $L_q u$ is quasi-periodic on the left 
can be replaced by their earliest form.
We can apply Algorithm \ref{algo_testQP} and \ref{earliestAlgo} on a part $q(x) u$
of a rule that is quasi-periodic by replacing $q(x) u$ by a new state $\hat{q}$.
Therefore each occurrence of $q(x) u$ in any rule is replaced by $\hat{q}(x)$ 
and for each rule $\ruleltw$ a rule 
$\hat{q}, f \rightarrow u_0 q_1(x_{\sigma(1)}) \dots q_n(x_{\sigma(n)}) u_n u$ 
is added. Then we can apply Algorithm \ref{algo_testQP} on $\hat{q}$ to test 
$q(x) u$ for quasi-periodicity on the left. 
If $\hat{q} u$ is quasi-periodic on the left we apply Algorithm \ref{earliestAlgo} 
on $\hat{q}$ to replace $q(x) u$ by their earliest form.

Second, we show that for any rule of an \ltw an equivalent rule can be constructed 
such that all parts $q(x) u$ in the rule where $L_q u$ is quasi-periodic on the left 
are earliest. Consider a rule \ruleltw.
Replacing a part $q(x) u$ with $L_q u$ is quasi-periodic by their earliest form as 
described above means that all occurrences of $\hat{q}(x)$ 
are replaced by $\lcp(\hat{q}) \hat{q}^e(x)$. Thus, to replace all parts of a rule 
that produce quasi-periodic languages by their earliest form the testing 
and replacing should be done from right to left as the earliest form may introduce 
new words on the left of the replaced state.
For a rule \ruleltw we start by testing $q_n(x_{\sigma(n)}) u_n$ for quasi-periodicity 
on the left and replace the part if necessary as described above. If so, we obtain 
$q, f \rightarrow u_0 q_1(x_{\sigma(1)}) \dots q_{n-1}(x_{\sigma(n-1)}) u_{n-1} \lcp(\hat{q_n}) \hat{q_n}^e(x_{\sigma(n)})$ 
and continue with testing $q_{n-1}(x_{\sigma(n-1)}) u_{n-1} \lcp(\hat{q_n})$ for 
quasi-periodicity on the left. If not, then we continue with testing 
$q_{n-1}(x_{\sigma(n-1)}) u_{n-1}$ for quasi-periodicity 
on the left. Following this construction for the rule from right to left, i.e. from index $n$ to $1$, 
leads to an equivalent rule where all parts $q(x) u$ with $L_q u$ is quasi-periodic on the left 
are earliest.

The construction runs in polynomial time as Algorithm \ref{algo_testQP} and \ref{earliestAlgo} 
run in polynomial time (for details see the proof of Theorem \ref{lem_noMoreQuasi}) and 
for each rule \ruleltw the algorithms are applied at most $n$ times.

\end{proof}

\section{Proof of Lemma \ref{lem_subtreeorder}}
\begin{lem}
 If $M$ and $M'$ are equivalent and in partial normal form then they are same-ordered.
\end{lem}
\begin{proof}
The proof is by contradiction.
Suppose that $M,\,M'$ are two equivalent \ltws in partial normal
form that are not same-ordered.
We consider co-reachable state $q$ and $q'$ 
of $M$ and $M'$ that are not same-ordered. Then for
$q,f \rightarrow u_0 q_1(x_{\sigma(1)}) \dots q_n(x_{\sigma(n)}) u_n$ and
$q',f \rightarrow v_0 q'_1(x_{\sigma'(1)}) \dots q'_n(x_{\sigma'(n)}) v_n$
we choose $i,j$, $i < j$ such that $j-i$ is minimal under the following constraints
\begin{compactitem}
 \item $\{\sigma(k)\ |\ i \leq k \leq j\} = \{\sigma'(k)\ |\ i \leq k \leq j\}$ and
 \item there is $k$, $i \leq k \leq j$ such that $\sigma(k) \neq \sigma'(k)$.
\end{compactitem}
As the set of indices of $\sigma(k)$ and $\sigma'(k)$ between $i$ and $j$ is the 
same but in different orders 
 there is $l$, $i \leq l \leq j$ such that $\sigma(l) > \sigma(l+1)$ or 
$\sigma'(l) > \sigma'(l+1)$. W.l.o.g.\ we assume that $\sigma(l) > \sigma(l+1)$. 
Additionally, from the above constraints it follows that there is $s, t$, $i \leq s \leq l < t \leq j$ 
such that $\sigma(s) = \sigma'(t)$. 
As the partial normal form satisfies the preconditions of Theorem~\ref{thm_periodic} 
we get that $q_s, \dots, q_t$ are periodic 
with the same period and $u_s, \dots, u_{t-1}$ are empty. 
Thus, $L_{q_l} u_l L_{q_{l+1}}$ is quasi-periodic. 
Then it follows from the partial normal form that $\sigma(l) < \sigma(l+1)$, a contradiction.

\end{proof}

\section{Proof of Theorem \ref{thm_main}}
\begin{thm}
 The equivalence problem for linear tree-to-word transducers is decidable in polynomial time.
\end{thm}
\begin{proof}
 Let $M$ and $M'$ be two linear tree-to-word transducers. 
 We construct equivalent \ltws $M_1$ and $M'_1$, respectively, such that 
 $M_1$ and $M_2$ are in partial normal form following Theorem \ref{thm_normalForm}.
 We then test if $M_1$ and $M'_1$ are same-ordered. 
 If they are same-ordered we test $M_1$ and $M'_1$ for equivalence, see Theorem~\ref{thm_eqSameOrder}. 
 If $M_1$ and $M'_1$ are not same-ordered we know following Lemma \ref{lem_subtreeorder} 
 that $M_1$ and $M'_1$ are not equivalent and therefore $M$ and $M'$ are not equivalent.
\end{proof}

\end{appendix}

\end{document}